\documentclass[pra,reprint,superscriptaddress,longbibliography,floatfix,twocolumn, nofootinbib,notitlepage]{revtex4-1}
\usepackage[utf8]{inputenc}
\usepackage{graphicx}
\usepackage[normalem]{ulem}
\usepackage{diagbox,hyperref}
 \usepackage{inconsolata}
\usepackage{amsmath,amssymb,amsthm,bm,mathtools,amsfonts,mathrsfs,bm, bbm,dsfont}
\usepackage{physics,xcolor}
\usepackage[shortlabels]{enumitem}
\usepackage{thmtools}

\usepackage{caption}
\usepackage{subcaption}

\newcommand{\be}{\begin{equation}}
\newcommand{\ee}{\end{equation}}

\newcommand{\id}{\mathds{1}}

\usepackage{ragged2e}

\renewcommand{\leq}{\leqslant}
\renewcommand{\geq}{\geqslant}

\newtheorem{theorem}{Theorem}
\newtheorem{definition}{Definition}

\newtheorem{proposition}{Proposition}

\begin{document}

\title{Equivalence between simulability of high-dimensional measurements and high-dimensional steering}

\author{Benjamin D.M. Jones}
\thanks{These authors contributed equally to this work.}
\affiliation{Department of Applied Physics, University of Geneva, 1211 Geneva, Switzerland}
\affiliation{H. H. Wills Physics Laboratory, University of Bristol, Bristol, BS8 1TL, UK}\affiliation{Quantum Engineering Centre for Doctoral Training,  University of Bristol, Bristol, BS8 1FD
UK}
\author{Roope Uola}
\thanks{These authors contributed equally to this work.}
\affiliation{Department of Applied Physics, University of Geneva, 1211 Geneva, Switzerland}
\author{Thomas Cope}
\affiliation{Institut für Theoretische Physik, Leibniz Universität Hannover,
30167 Hannover, Germany}
\author{Marie Ioannou}
\affiliation{Department of Applied Physics, University of Geneva, 1211 Geneva, Switzerland}
\author{S\'ebastien Designolle}
\affiliation{Department of Applied Physics, University of Geneva, 1211 Geneva, Switzerland}
\author{Pavel Sekatski}
\affiliation{Department of Applied Physics, University of Geneva, 1211 Geneva, Switzerland}
\author{Nicolas Brunner}
\affiliation{Department of Applied Physics, University of Geneva, 1211 Geneva, Switzerland}

\begin{abstract}
    The effect of quantum steering arises from the judicious combination of an entangled state with a set of incompatible measurements. Recently, it was shown that this form of quantum correlations can be quantified in terms of a dimension, leading to the notion of genuine high-dimensional steering. While this naturally connects to the dimensionality of entanglement (Schmidt number), we show that this effect also directly connects to a notion of dimension for measurement incompatibility. More generally, we present a general connection between the concepts of steering and measurement incompatibility, when quantified in terms of dimension. From this connection, we propose a novel twist on the problem of simulating quantum correlations. Specifically, we show how the correlations of certain high-dimensional entangled states can be exactly recovered using only shared randomness and lower-dimensional entanglement. Finally, we derive criteria for testing the dimension of measurement incompatibility, and discuss the extension of these ideas to quantum channels.
\end{abstract}

\maketitle

\section{Introduction} 

High-dimensional quantum systems feature a number of interesting phenomena, beyond what is possible for qubit systems. For example, the effect of entanglement is known to become increasingly robust to noise when higher dimensions are considered, the robustness becoming even arbitrary large \cite{zhu2021high, ecker2019overcoming}. In turn, the nonlocal correlations obtained from measurements on high-dimensional systems also feature significantly increased robustness. Indeed, these effects offer interesting perspectives for quantum information processing, allowing, e.g., for quantum communications over very noisy channels.

In this work, we consider the effect of genuine high-dimensional steering (GHDS), which has been introduced recently \cite{designolle2021genuine}. The steering scenario can be viewed as the certification of entanglement between an untrusted party (Alice) and a trusted one (Bob). Hence steering is usually referred to as being one-sided device-independent (1-SDI). The key point of GHDS is to certify not only the presence of entanglement, but a minimal dimensionality of entanglement (specifically the Schmidt number) from observed correlations in a 1-SDI scenario. More formally, this approach introduces the notion of $n$-preparable assemblages, i.e., those assemblages being preparable based on any possible entangled state of Schmidt rank at most $n$; 1-preparable assemblages being then simply those assemblages that cannot lead to steering. Next, one can construct a steering inequality for $n$-preparable assemblages, the violation of which implies the presence of genuine $n+1$-dimensional steering. This was demonstrated in a quantum optics experiment (based on photon-pairs entangled in orbital angular momentum) reporting the 1-SDI certification of 14-dimensional entanglement.

A natural question at this point is to understand what are the resources required in terms of measurements for demonstrating GHDS. Indeed, the effect of steering uses not only an entangled state as a resource, but also a well-chosen set of local measurements for Alice. The latter must be incompatible (in the sense of being non-jointly measurable), but it turns out that steering has a direct connection to measurement incompatibility. 

The present work explores this question, and establishes a general connection between GHDS and the notion of $n$-simulability of high-dimensional measurements which has been recently introduced in Ref.~\cite{ioannou2022simulability}. This notion generalises the concept of joint measurability and provides a quantification of measurement incompatibility in terms of a dimension. The connection we uncover generalises the well-known relations between quantum steering and joint measurability. Moreover, we also extend the connection to quantum channels, in particular the characterisation of their high-dimensional properties. These general tripartite connections between high-dimensional steering, measurements and channels, allow for results of one area to be directly translated in others, which we illustrate with several examples. 

\begin{figure*}[t]
    \centering
    \includegraphics[width=0.99\textwidth]{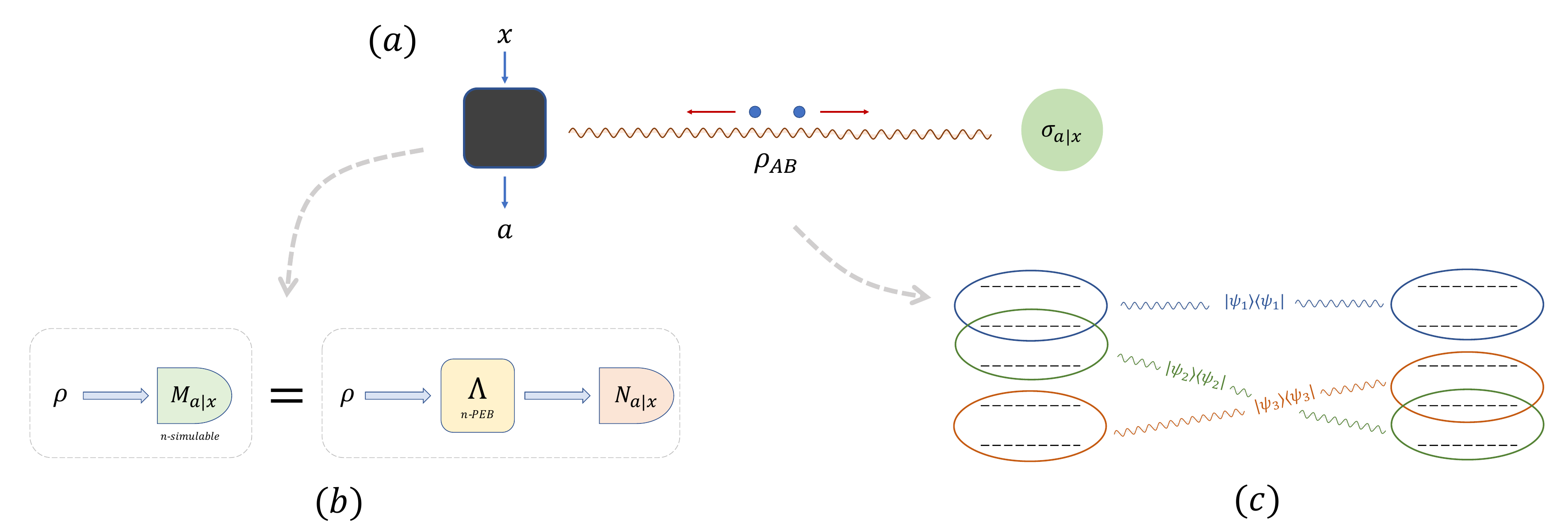}
    \caption{Concepts and connections that appear in this work. (a) Quantum Steering scenario. (b) A set of measurements is $n$-simulable if they can be replaced by an $n$-partially entanglement breaking channel ($n$-PEB) followed by some measurements. (c) Illustration of the Schmidt number (SN) of a bipartite state: the state of two $5$ level systems is a combination of states with only qubit entanglement, hence the overall state has SN at most $2$. }
    \label{fig:fig1}
\end{figure*}

\section{Summary of results}

We start by identifying the resources for GHDS. In particular, we show that an assemblage is $n$-preparable if it can be prepared via an entangled state of Schmidt number $n$ or if the set of Alice's local measurements are $n$-preparable. Hence the observation of genuine $n+1$-dimensional steering implies the presence of both (i) an entangled state of Schmidt number (at least) $n+1$, and (ii) a set of measurements for Alice that is not $n$-simulable. In this sense, GHDS provides a dimensional certification of both the entangled state and the local measurements. Moreover, we show that there is a one-to-one mapping between any $n$-preparable assemblage and a set of measurements that is $n$-simulable, generalising the existing connection between steering and joint measurability (corresponding here to the case $n=1$).

This connection allows us to import results from one area to the other. For example, we can construct optimal models for simulating the correlations of high $d$-dimensional entangled states (so-called isotropic states) based on lower $n$-dimensional entanglement (and classical shared randomness). This simulation models hold for all possible local measurements on Alice's and Bob's side. In this sense, these models can be considered as a generalisation of the well-known local hidden state model of Werner, where classical shared randomness is augmented with low-dimensional entanglement. Moreover, we can translate steering inequalities for GHDS into criteria for testing non $n$-simulability of measurements. 

Finally, we obtain a dimensional characterisation of quantum channels via channel-state duality. In particular, we consider channels that map the set of all measurements to $n$-simulable ones, and describe the corresponding Choi states.

We conclude with a number of open questions.

\section{Basic concepts and questions}

A central notion for us will be \textit{quantum steering}, see e.g.  \cite{cavalcanti2016quantum, uola2020quantum} for recent reviews.  Here, one party (Alice) performs local measurements $\{M_{a|x}\}$ on a state $\rho_{AB}$, a unit-trace positive semi-definite matrix acting on a finite-dimensional Hilbert space, that she shares with another distant party (Bob). The measurements are collections of matrices for which $M_{a|x} \geq 0$ $\forall a,x$ and $\sum_a M_{a|x} = \id $ for each $x$. Here $x$ indexes the measurement and $a$ indexes the outcome. For each $x$, the collection $\{M_{a|x}\}_a$ is called a positive operator-valued measure (POVM for short). By performing her measurements, Alice remotely prepares the system of Bob in different possible states denoted by 
\be
\sigma_{a|x}:=\text{Tr}_A \bigg (M_{a|x}\otimes \id ~ [\rho_{AB}] \bigg), \label{eq:steeringassem}
\ee
usually termed an \textit{assemblage}, see Figure (\ref{fig:fig1}a). Such an assemblage demonstrates quantum steering when it does not admit a \textit{local hidden state} (LHS) model, i.e. a decomposition of the form 
\begin{equation} \label{LHS}
\sigma_{a|x} = p(a|x) \sum_\lambda ~ p(\lambda|a,x) ~ \sigma_\lambda \,,
\end{equation}
where $p(a|x)$ is a normalisation factor and $p(\lambda|x,\lambda)\sigma_\lambda$ is an ensemble of states whose priors get updated upon Bob asking Alice to perform the measurement $x$ and her reporting back the outcome $a$. 

Steering represents a form of quantum correlations that is intermediate between entanglement and Bell nonlocality \cite{wiseman2007steering,quintino2015inequivalence}. Specifically, there exist entangled states that cannot lead to steering, and there exist some steerable states that cannot lead to Bell inequality violation (nonlocality). Also, the steering scenario is commonly referred to as \textit{one-sided device-independent} (1-SDI), as Alice's device is untrusted but Bob's device is fully characterised. Since steering requires the presence of entanglement---separable states always admitting an LHS model---it also represents a 1-SDI method for certifying entanglement. Moreover, steering is an asymmetric phenomenon, as there exist states $\rho_{AB}$ for which steering is only possible in one direction (e.g., from $A$ to $B$) \cite{bowles2014one}.

One can take the concept of quantum steering a step further in terms of bipartite entanglement detection. Instead of only certifying the presence of entanglement, it is possible to use steering to characterise the dimensionality entanglement dimensionality, as quantified via the Schmidt number \cite{terhal2000schmidt}. For a pure state $\ket{\psi}$, this corresponds to the \textit{Schmidt rank} (SR), i.e., the minimum number of terms needed to express $\ket{\psi}$ as a linear combination of product states. The \textit{Schmidt number} (SN) \cite{terhal2000schmidt} is a generalisation to mixed states, formally defined as 
\begin{align}
    \text{SN}(\rho) := \underset{\{ p_k, ~\ket{\psi_k} \}}{\min} \max_k \quad&\text{SR}(\ket{\psi_k}) \\\nonumber
    \quad \text{s.t} \quad &\rho = \sum_k p_k \ketbra{\psi_k}{\psi_k}.
\end{align}
The Schmidt number thus quantifies the entanglement dimensionality, in that it tells the minimum number of degrees of freedom that one needs to be able to entangle in order to produce the state, see Figure (\ref{fig:fig1}c). As an example, witnessing a Schmidt number of three implies that qubit entanglement, even when mixed between different subspaces, is not enough to produce the state.

In \cite{designolle2021genuine}, the concept of \textit{genuine high-dimensional steering} (GHDS) was introduced, where one asks whether a given assemblage $\sigma_{a|x}$ can be produced using a bipartite state $\rho_{AB}$ of Schmidt number at most $n$, in which case we term the assemblage \textit{$n$-preparable}. In this framework, an assemblage is LHS if and only if it is $1$-preparable, as any LHS assemblage can be prepared using only separable states \cite{Kogias15,Moroder16}. Hence if an assemblage is not $n$-preparable, this guarantees that the underlying state $\rho_{AB}$ is of Schmidt number at least $n+1$. This represents a 1-SDI certification of entanglement dimensionality, illustrated in a recent quantum optics experiment certifying up to 14-dimensional entanglement \cite{designolle2021genuine}.

So far, the focus of GHDS is on the dimensionality of the shared entangled state. There is however another resource that is crucial for observing quantum steering, namely the set of measurements performed by Alice, which must be incompatible. More generally, there exist in fact a deep connection between measurement incompatibility (in the sense of being not jointly measurable) and quantum steering \cite{quintino14,uola14,uola15}. In particular, this implies that any set of incompatible measurements for Alice can be combined with an appropriate state $\rho_{AB}$ for demonstrating steering. 

This naturally raises the question of what are the necessary resources in terms of measurements for demonstrating GHDS. Intuitively, the latter should also require a minimal ``dimensionality'' for the set of measurements. Below we will make this intuition precise, by using the concept of $n$-simulability of a set of measurements. More generally, we will establish a deep connection between GHDS (more precisely the notion of $n$-preparability of an assemblage) and $n$-simulability of set of measurements. This generalises the previously known connection between steering and measurement incompatibility. 

A set of measurements $\{M_{a|x}\}$, defined on a Hilbert space of dimension $d$, is said to be $n$-simulable when the statistics of this set of measurements on any possible quantum state can be exactly recovered using a form of compression of quantum information to a lower $n$-dimensional space. Consider for example Alice (on the moon), sending an arbitrary state $\rho$ to a distant party Bob (on earth), who will perform a set of POVMs $\{M_{a|x}\}$ (see Fig. 1). Which POVM Bob performs depends on some input $x$. The expected (target) data is given by $p(a|x,\rho) = \Tr(M_{a|x} \rho)$. As resource, we consider here the dimensionality of the quantum channel between Alice and Bob, while a classical channel is always available for free. The goal is then to compress as much as possible the initial state of Alice, in order to use a quantum channel with minimal dimension, while still recovering exactly the target data. More formally, we demand that
\begin{equation} \label{n-simulable}
    M_{a|x} = \sum_{\lambda}  \Lambda_{\lambda}^*( N_{a|x,\lambda})
\end{equation}
where $\Lambda = \{\Lambda_{\lambda}\}_{\lambda}$ denotes the instrument (compressing from dimension $d$ to $n$), with classical output $\lambda$, and $N_{a|x,\lambda}$ is a set of $n$-dimensional POVMs performed by Bob upon receiving the input $x$ and the classical information $\lambda$ communicated by Alice. Here $\Lambda_\lambda^*$ refers to the Heisenberg picture of $\Lambda_\lambda$. A set of measurements is termed $n$\textit{-simulable} whenever a decomposition of the form \eqref{n-simulable} can be found.

An important case is $1$-simulability, i.e., when the full quantum information can be compressed to purely classical one. This is possible if and only if the set of POVMs is jointly measurable, i.e., 
$M_{a|x} = \sum_\lambda ~ p(a|x,\lambda) ~ G_\lambda$, for some probability distribution $p(a|x,\lambda)$ and a ``parent'' measurement $G_\lambda$, see \cite{JMinvitation,JMreview} for reviews on the topic. A set of POVMs that is not jointly measurable (hence called \textit{incompatible}), can nevertheless be $n$-simulable, for some $n$ with $2 \leq n \leq d$. 

The notion of $n$-simulability can also be connected to quantum channels, and their dimensional properties. This requires the use of a property of channels that is analogous Schmidt number of bipartite states. Namely, one says that a channel $\Lambda$ is \textit{$n$-partially entanglement breaking} ($n$-PEB) if $\text{SN}(\Lambda \otimes \id \rho) \leq n$ for all $\rho$ \cite{chruscinski2006partially}. Clearly, for the case $n=1$ this concept corresponds to entanglement breaking channels.

This leads to an alternative formulation of \textit{$n$-simulability}, which we will primarily use in the following sections: a measurement assemblage $M_{a|x}$ is \textit{$n$-simulable} if and only if there exists an $n$-PEB quantum channel $\Lambda$ and a measurement assemblage $N_{a|x}$ such that $M_{a|x} = \Lambda^* \big ( N_{a|x} \big )$.

In the rest of the paper, we will first establish precisely the connection between $n$-preparability and $n$-simulability. In turn, we will discuss simulation models for for the correlations of entangled states (of Schmidt number $d$) using as resource lower-dimensional entanglement (of Schmidt number $n<d$), considering all possible measurements. This idea can be seen as a generalisation of the problem of simulating the correlations of entangled state via local hidden variables (or local hidden state models). Finally, in the last section of the paper, we will also extend the connection to quantum channels and their characterisation in terms of dimension. This will provide a full tripartite connection, for characterising dimension in steering assemblages, incompatibility of sets of measurements, and quantum channels.

\section{High-dimensional steering and simulability of measurements} 
In this section, we present in detail the structural connection between $n$-preparability of steering assemblages and $n$-simulability of sets of measurements. 

We start with a first result clearly identifying the resource for GHDS. More precisely, the following Theorem implies that observing GHDS, i.e., an assemblage which is not $n$-preparable, implies that (i) the shared entangled state $\rho_{AB}$ has at least Schmidt number $n+1$, and (ii) the set of measurements $\{M_{a|x}\}$ performed by Alice is not $n$-simulable. In other words, one really needs both high-dimensional entanglement and high-dimensional measurement incompatibility to witness genuine high-dimensional steering.

More formally we can prove the following.

\begin{theorem} \label{theorem:compat->prep}
If $M_{a|x}$ is $n$-simulable or $\rho_{AB}$ has Schmidt number at most $n$, then the assemblage
\be
\sigma_{a|x}:=\textup{Tr}_A \bigg (M_{a|x}\otimes \id ~ [\rho_{AB}] \bigg)
\ee
is $n$-preparable.
\end{theorem}

\begin{proof}
If $\rho_{AB}$ has SN at most $n$, this simply follows from the definition of $n$-preparability. Now suppose that $M_{a|x}$ is $n$-simulable. Then there exists a $n$-PEB channel $\Lambda$ and measurements $N_{a|x}$ such that $M_{a|x} = \Lambda^* (N_{a|x})$. By the definition of the dual, we can hence write
\begin{align}
\sigma_{a|x} &= \Tr_A \Big ( \Lambda^* (N_{a|x}) \otimes \id [ \rho_{AB} ] \bigg ) \\
& = \Tr_A \Big ( \big(N_{a|x} \otimes \id\big) \big( \Lambda \otimes \id  \big) [\rho_{AB}] \Big )
\end{align}
and as $\Lambda$ is $n$-PEB, then $\Lambda \otimes \id [\rho_{AB}]$ has SN at most $n$, so $\sigma_{a|x}$ is $n$-preparable.
\end{proof}

It is worth noting that, for the simplest case of $n=1$, the above Theorem corresponds to the well-known fact that an assemblage constructed from a separable state or via a jointly measurable set of POVMs always admits a LHS model. In other words, the observation steering proves the presence of an entangled state and an incompatible set of POVMs for Alice.

Our next result establishes a general equivalence between any $n$-preparable assemblage and a set of POVMs that is $n$-simulable, and vice versa. The main idea is that a set of quantum measurements $M_{a|x}$ and a steering assemblage $\sigma_{a|x}$ are very similar types of mathematical objects: both are composed of positive semi-definite matrices, and $\sum_a M_{a|x}=\id\quad \forall x$ whereas $\sum_a \sigma_{a|x}$ will be equal to some fixed state $\rho_B = \text{Tr}_A (\rho_{AB})$ for all $x$. A direct connection can be established, namely that $\sigma_{a|x}$ is LHS if and only if $\rho_B^{-\frac{1}{2}} \sigma_{a|x} \rho_B^{-\frac{1}{2}}$ is jointly measurable (when interpreted as a set of measurements) \cite{uola15}. The Theorem below can be considered a generalisation of this result, in the sense that the proof of Ref. \cite{uola15} corresponds to the case $n=1$.

\begin{theorem} \label{thm: n-sim = n-prep}
Consider a steering assemblage $\sigma_{a|x}$ and measurements $M_{a|x}$ such that $M_{a|x}=\rho_B^{-\frac{1}{2}} ~ \sigma_{a|x} ~ \rho_B^{-\frac{1}{2}}$, where $\rho_{B} := \sum_a \sigma_{a|x}$ is of full rank. Then $M_{a|x}$ is $n$-simulable if and only if $\sigma_{a|x}$ is $n$-preparable.
\end{theorem}

\begin{proof}  Let $N_{a|x}$ be a measurement assemblage and $\rho_{AB}$ be a state such that $\Tr_A(\rho_{AB}) = \rho_B$. Let $(\cdot)^T$ denote the transpose with respect to an eigenbasis of $\rho_B$. We then have the following equivalences
\begin{align}
\sigma_{a|x} &= \Tr_A (N_{a|x} \otimes \id~\rho_{AB} ) \\
\iff M_{a|x}&= ~ \rho_B^{ -\frac{1}{2}} ~\Tr_A (N_{a|x} \otimes \id~\rho_{AB} ) ~ \rho_B^{ -\frac{1}{2}} \\ 
\iff M_{a|x}^T&= ~ \rho_B^{ -\frac{1}{2}} ~\Tr_A (N_{a|x} \otimes \id~\rho_{AB} )^T ~ \rho_B^{ -\frac{1}{2}} \\ 
\iff M_{a|x}^T &= \Lambda_{\rho_{AB}}^* \Big ( N_{a|x} \Big ),
\end{align}
where in the third line we used the fact that $(\rho_B^{-\frac{1}{2}})^T=\rho_B^{-\frac{1}{2}}$, as the transpose is taken in an eigenbasis of $\rho_B$, and in the last line we have invoked the form of channel-state duality from Ref.~\cite{kiukas2017continuous}. 

Now observe that the existence of a state $\rho_{AB}$ in the above with Schmidt number at most $n$ is equivalent to $\sigma_{a|x}$ being $n$-preparable. We can also see that there exists $\rho_{AB}$ with SN$(\rho_{AB})\leq n$ if and only if $M_{a|x}^T$ is $n$-simulable, as such state corresponds to $\Lambda_{\rho_{AB}}$ being $n$-PEB, see Appendix A for details. To finalize the proof we must show that $M_{a|x}$ is $n$-simulable if and only if $M_{a|x}^T$ is $n$-simulable. This can be seen as follows. First note that $M_{a|x}^T$ defines a valid collection of measurements. Suppose that $M_{a|x} = \Lambda^*(N_{a|x})$ with $\Lambda$ $n$-PEB and $N_{a|x}$ arbitrary measurements. Then letting $\mathcal{T}$ denote the transpose map, we have that $M_{a|x}^T = (\mathcal{T} \circ \Lambda^*)(N_{a|x}) = ( \Lambda \circ \mathcal{T}^*)^*(N_{a|x})$. As $\Lambda$ is $n$-PEB, $\Lambda \circ \mathcal{T}^*$ is also $n$-PEB. Hence $M_{a|x}^T$ is $n$-simulable. The converse direction follows from $(M_{a|x}^T)^T = M_{a|x}$.
\end{proof}

As a technical remark, note that as for any $a$ and $x$ the support of $\sigma_{a|x}$ is contained within the support of $\rho_B = \sum_a \sigma_{a|x}$ (this follows as $\sigma_{a|x}$ are all positive semi-definite), we can still invoke the above theorem in the case where $\rho_B$ is not full rank, by restricting $\sigma_{a|x}$ to the support of $\rho_B$.

Theorem \ref{thm: n-sim = n-prep} also allows to prove the following result, which complements Theorem \ref{theorem:compat->prep}. This shows that for any set of POVMs that is not $n$-simulable, one can always find an entangled state such that the resulting assemblage is not $n$-preparable. Again, this generalizes some previous results stating that any incompatible set of POVMs can lead to steering \cite{quintino14,uola14}, which corresponds to the case $n=1$ of the proposition below.

\begin{proposition}
If $M_{a|x}$ is not $n$-simulable, then the assemblage
\be
\sigma_{a|x}:=\textup{Tr}_A \bigg (M_{a|x}\otimes \id ~ \ketbra{\Phi^+} \bigg)
\ee
is not $n$-preparable, where $\ket{\Phi^+}=\frac{1}{\sqrt{d}}\sum_i \ket{ii}$.
\end{proposition}
\begin{proof}
We have that
\begin{align}
\sigma_{a|x}=\text{Tr}_A \bigg (M_{a|x}\otimes \id ~ \ketbra{\Phi^+} \bigg) = \frac{1}{d}~ M_{a|x}^T.
\end{align}
By the proof of Theorem~\ref{thm: n-sim = n-prep}, if $M_{a|x}$ is not $n$-simulable, then $M_{a|x}^T$ is not $n$-simulable. Then invoking Theorem \ref{thm: n-sim = n-prep} with $\rho_B = \frac{\id}{d}$, we have that that $\sigma_{a|x}$ is not $n$-preparable.
\end{proof}

In the final part of this section, we show that the trade-off between high-dimensional entanglement, high-dimensional measurement incompatibility, and high-dimensional steering can be made quantitative. For this, we use a specific resource quantifiers known as the convex weight \cite{Steeringweight}. Consider for example the quantification of entanglement via the weight. For any entangled state $\rho$, we can measure its entanglement through its weight, given by the following quantity
\begin{equation}
\label{eq: WeightDef}
\begin{split}
    \mathcal{W}_F(\rho) := &\min \lambda\\
    &\mathrm{s.t.}\  D= (1-\lambda) \rho_{sep} + \lambda\sigma,
\end{split}
\end{equation}
where the minimisation runs over any state $\rho_{sep}$ that is separable, and $\sigma$ an arbitrary state. As expected, $\mathcal{W}_F (\rho) =0$ when $\rho$ is separable. More generally, this quantifier can apply to objects such as states, measurements or steering assemblages, with respective free sets $E_n$: the set of states with Schmidt number at most $n$, $S_n$: the set of of $n$-simulable measurements assemblages, and $P_n$: the set of $n$-preparable steering assemblages. We can now state our next result, which quantitatively illustrates the necessity of high-dimensional measurement incompatibility and entanglement for GHDS:

 \begin{restatable}{theorem}{weighttheorem} \label{thm:weight} Given an assemblage $\sigma_{a|x}=\textup{Tr}_A  (M_{a|x}\otimes \id ~ [\rho_{AB}])$, we have the following inequality:
\[
\mathcal{W}_{P_n}(\sigma_{a|x}) \leq \mathcal{W}_{S_n}(M_{a|x}) \mathcal{W}_{E_n}(\rho_{AB}).
\]
For the case $n=1$ we get a quantitative connection among steering, measurement incompatibility and entanglement.
\end{restatable}

We defer the proof of this theorem to the appendix.

\section{Simulating the correlations of high-dimensional entangled states using low-dimensional entanglement}

Strong demonstrations of the non-classical correlations of entangled states comes from the observation of Bell inequality violation, or from quantum steering. A long-standing topic of research is to understand the link between entanglement and these stronger forms of quantum correlations, see e.g. \cite{brunner2014bell,Augusiakreview}. In a seminal paper, Werner showed that certain entangled states, referred to as Werner states, cannot lead to Bell inequality violation \cite{Werner1989}. This result is based on the construction of an explicit local hidden variable model that reproduces exactly the correlations expected from any possible local projective measurements on the Werner state. Moreover, it turns out that the model construct by Werner is in fact of the form of an LHS model (as in Eq. \eqref{LHS}, see also Fig~\ref{fig:LHSmodel}), hence these Werner states can also never lead to quantum steering \cite{wiseman2007steering}. Note that these results can be extended to general POVMs using the model of Ref. \cite{barrett2002nonsequential}, which can be shown to be of LHS form \cite{quintino2015inequivalence}.

Here we revisit the above questions and propose a new perspective, based on the ideas developed in the previous sections of the paper. Instead of considering simulation models that involve only classical resources (classical shared randomness), we consider now simulation models assisted by entanglement, see Fig.~\ref{fig:EALHSmodel}. Of course, for this problem to be non-trivial, we must demand that the entanglement used in the simulation model is somehow weaker compared to the entanglement 
of the original state to be simulated. The dimensionality of entanglement (as given by the Schmidt number) provides a good measure for this problem. 

Consider an entangled state $\rho_{AB}$ of Schmidt number $d$ and arbitrary local measurements (possibly infinitely many) for both Alice and Bob. We now ask if we can simulate the resulting correlations with a model involving lower-dimensional entangled states (of Schmidt number $n<d$) and classical shared randomness. Of course, building such models can be challenging, as the model should reproduce exactly all correlations for any possible choice of local measurements. Nevertheless, we will see that using the ideas developed above, we can come up with such entanglement-assisted simulation models, and moreover prove their optimality. 

The main idea to construct these simulation models is to apply Theorem~\ref{theorem:compat->prep} to a result obtained recently in \cite{ioannou2022simulability}. The latter consist in obtaining bounds (in terms of noise robustness) for $n$-simulability for the (continuous) set of all projective measurements (in dimension $d$) under white noise. From Theorem~\ref{theorem:compat->prep}, we obtain an equivalent assemblage (with a continuous input $x$) that is $n$-preparable. The last point is to notice that this assemblage corresponds in fact to the one obtained from performing arbitrary local projective measurements on a shared entangled state $\rho_{AB}$, which takes the form of an isotropic state, i.e. 
\begin{equation} \label{iso}
    \rho(\eta'):= \eta' \ketbra{\Phi^+} + (1-\eta') \frac{\id}{d^2}
\end{equation}
where $\ket{\Phi^+}=\frac{1}{\sqrt{d}}\sum_i \ket{ii}$ and $0 \leq \eta' \leq 1$. Hence we obtain a simulation model using only entanglement with Schmidt number $n$ which reproduces exactly the correlations of some isotropic state of dimension $d\times d$. Interestingly, it appears that this isotropic state can have a Schmidt number that is larger than $n$.

More formally, consider the set of all projective measurements (PVMs) subject to white noise

\begin{equation}\label{noisyPVM}
    \mathcal{M}_{PVM}^\eta:=\bigg \{\eta M_{a|U} + (1-\eta)\frac{\id}{d} ~ : ~ U\in U(d) \bigg \} \,,
\end{equation}
where $U(d)$ is the unitary matrix group, $M_{a|U} = U\ketbra{a}U^\dagger$ and $\ket{a}$ denotes the computational basis.
It was shown in \cite{ioannou2022simulability} that the set $\mathcal{M}_{PVM}^\eta$ is $n$-simulable if $\eta \leq (d \sqrt{\frac{n+1}{d+1}}-1)(d-1)^{-1}$. Then by passing the noise from the measurements onto the state (see for example \cite{uola14}), we have that:
\begin{align}
&\text{Tr}_A \bigg (\bigg [\eta M_{a|U} + (1-\eta)\frac{\id}{d} \bigg ]\otimes \id ~ \ketbra{\Phi^+})\\
=&\text{Tr}_A \bigg (M_{a|U}\otimes \id ~ \rho (\eta) \bigg ).
\end{align}
Hence we reproduce exactly the assemblage expected from arbitrary projected measurement on an isotropic state with $\eta'= \eta$. Moreover, it is known that $\text{SN}(\rho(\eta)) \geq n+1 \quad \text{if} \quad \eta > \frac{dn-1}{d^2-1}$ \cite{terhal2000schmidt}. Hence for
\be
\frac{dn-1}{d^2-1} < \eta \leq \frac{d \sqrt{\frac{n+1}{d+1}}-1}{d-1}
\ee
the resulting assemblage can be reproduced via a simulation model involving only entangled states of Schmidt number $n$, despite the state possessing a Schmidt number of $n+1$. More generally, one can deduce a general bound on the noise parameter $\eta$ for guaranteeing $n$-preparability. We have illustrated these bounds this in Fig.~\ref{fig:my_label} for the case of dimension four. Remarkably, as the construction for PVMs in Ref. \cite{ioannou2022simulability} is optimal, the simulation models we obtain are also optimal (considering all possible PVMs). An interesting question is to understand how to extend these bounds considering all POVMs, but this is a challenging question, still open for the simplest case of $n=1$.

\begin{figure}
     \centering
     \begin{subfigure}[b]{0.5\textwidth}
         \centering
         \includegraphics[width=\textwidth]{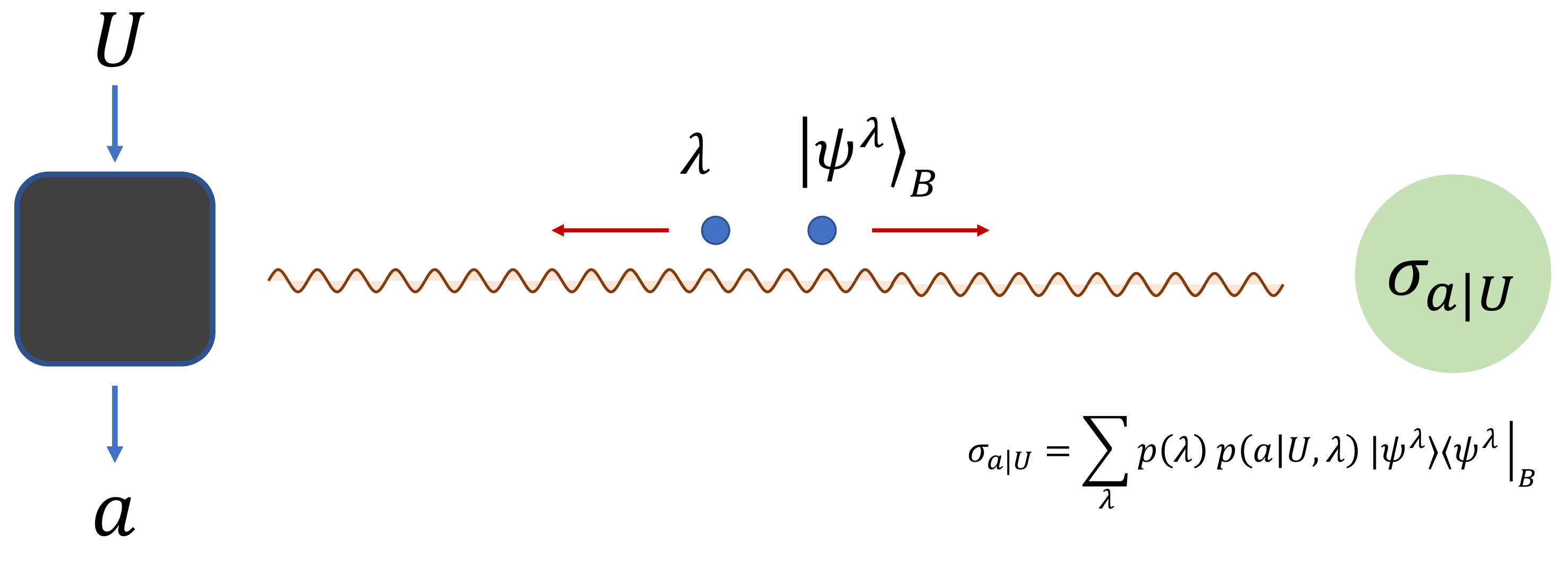}\vspace{-10pt}
         \caption{}
         \label{fig:LHSmodel}
     \end{subfigure}
     \hfill
     \begin{subfigure}[b]{0.5\textwidth}
         \centering
         \vspace{5pt}\includegraphics[width=\textwidth]{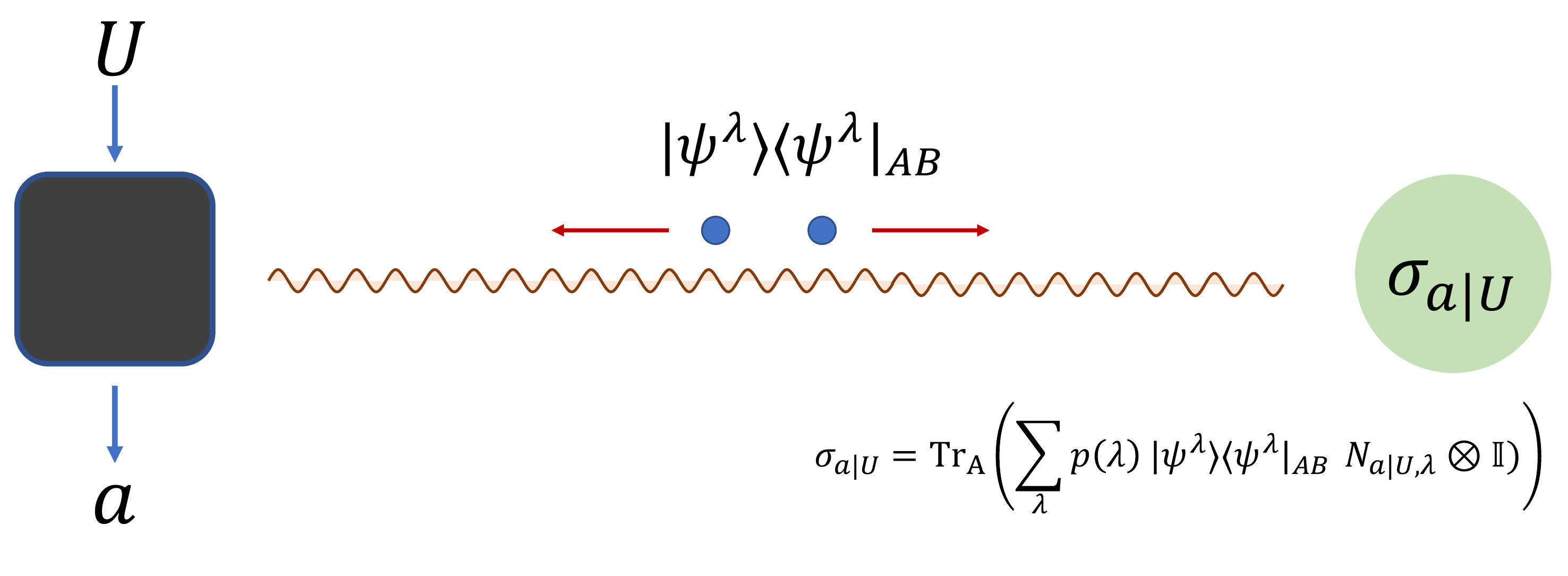} \vspace{-10pt}
         \caption{}
         \label{fig:EALHSmodel}
     \end{subfigure}
     \hfill
     \begin{subfigure}[b]{0.5\textwidth}
    \centering
    \hspace*{-5pt}\includegraphics[width=\textwidth]{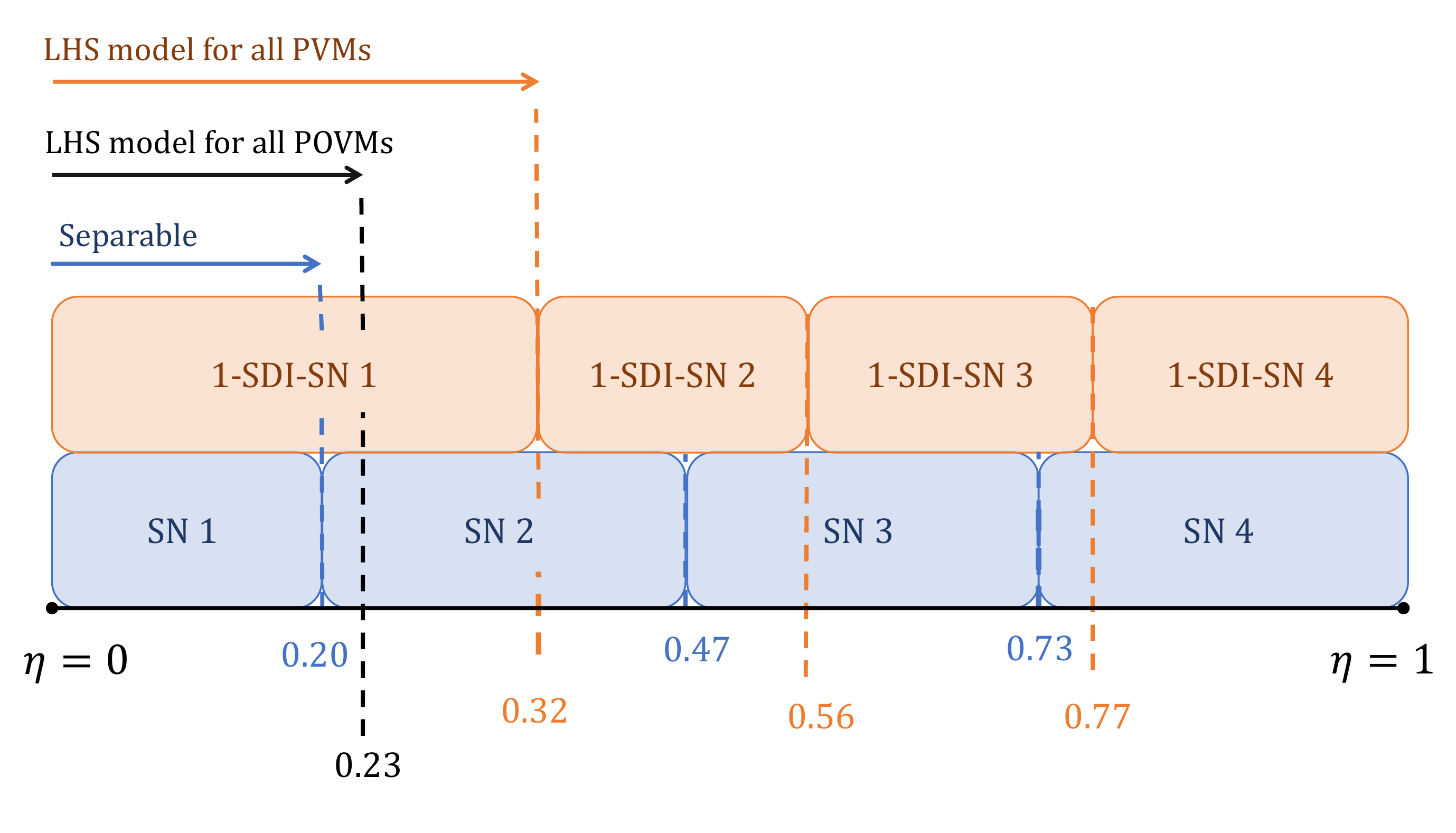} \vspace*{-15pt}
    \caption{}
    \label{fig:my_label}
     \end{subfigure}
     \caption{(a) Local hidden state model: one aims at simulating the assemblage with a separable state.\\
     (b) $n$-preparability: simulation of an assemblage using states with low-dimensional entanglement.\\
     (c) High-dimensional entanglement and steering properties of the isotropic state with local dimension $d=4$. The Schmidt number (SN) bounds can be found in \cite{terhal2000schmidt}, and in this work we translate known values on the $n$-simulability of all PVMs from \cite{ioannou2022simulability} to thresholds on states being such that they can only lead to n-preparable assemblages under all projective measurements. In the figure this is referred to as the one-sided semi-device independent Schmidt number (1-SDI-SN) under all PVMs. The bound for LHS models for all POVMs is from \cite{barrett2002nonsequential, almeida2007noise}.}
\end{figure}

\section{Criteria for $n$-simulability} 

The connections established in Section III also allow us to translate $n$-preparability inequalities into criteria for $n$-simulability. As an example, we take the set of $n$-preparability witnesses presented in Ref. \cite{designolle2021genuine}. Such witnesses state that for an $n$-preparable state assemblage $\{\sigma_{a|x}\}$ with $2$ inputs and $d$ outputs, one has that 
\begin{equation} \label{witness}
    \sum_{a,x}\text{Tr}[\sigma_{a|x}W_{a|x}]\leq N\Big(\frac{\sqrt{n}-1}{\sqrt{n}+1}+1\Big) \,.
\end{equation} 
where $N=1+1/\sqrt{d}$. The witness $W_{a|x}$ consists of a pair of mutually unbiased bases (MUBs for short) transposed in the computational basis, i.e., $W_{a|1}=|a\rangle\langle a|$ and $W_{b|2}=|\varphi_b\rangle\langle\varphi_b|^T$, where $\{|a\rangle\}$ is the computational basis and $\{|\varphi_b\rangle\}$ is an orthonormal basis with the property $|\langle a|\varphi_b\rangle|^2=1/d$ for each $a$ and $b$.

As an $n$-simulable set of measurements leads to an $n$-preparable state assemblage by Theorem \ref{theorem:compat->prep}, violation of a witness of this type in a steering scenario verifies that Alice's measurements are not $n$-simulable. As an example, we take a pair of MUBs subjected to white noise with visibility $\eta$ (similarly to Eq. \eqref{noisyPVM}) on Alice's side and the isotropic state \eqref{iso}. Plugging the resulting assemblage into the witness \eqref{witness}, we get that
\begin{align}
    \eta\leq\frac{(d+\sqrt{d}-1)\sqrt{n}-1}{(d-1)(\sqrt{n}+1)}.
\end{align}
Hence, for a visibility larger than this bound, a pair of MUBs is provably not $n$-simulable. We note that for the case $n=1$ we retrieve the known tight joint measurability threshold of two MUBs subjected to white noise \cite{Carmeli12,Haa15,Uola16}. Obtaining similar bounds for complete sets of MUBs, known for $n=1$ \cite{Designolle2019}, would be interesting.

\section{Quantum channels} 

An important superset of entanglement breaking channels is that of incompatibility breaking channels \cite{heinosaari2015incompatibility}, which are channels $\Lambda$ such that $\Lambda^*(M_{a|x})$ is jointly measurable for any $M_{a|x}$. Via channel-state duality these channels correspond respectively to separable and unsteerable states (where the direction of unsteerability corresponds to whether the channel is applied on the first or second system in the definition of channel-state duality). The connections between high-dimensional steering, $n$-simulability and $n$-PEB channels motivate the following definition:

\begin{definition} \label{def: PSB}
A channel $\Lambda$ is \textbf{$n$-partially incompatibility breaking} ($n$-PIB) if for any measurement assemblage $N_{a|x}$ the resulting measurement assemblage $\Lambda^*(N_{a|x})$ is $n$-simulable\footnote{We note that our definition here is different to the notion of $n$-incompatibility breaking channels defined in \cite{heinosaari2015incompatibility}, which denotes channels who break the incompatibility of any $n$ observables.}.
\end{definition}

Hence, just as $\Lambda \otimes \id$ maps all bipartite states to states with Schmidt number $n$ for $\Lambda$ a $n$-PEB channel, an $n$-PIB channel maps any measurement assemblage to an $n$-simulable one (in the Heisenberg picture).  We can also gain insight from considering the structure of $n$-PIB channels and their relation to $n$-PEB channels. Elaborating upon Def.~\ref{def: PSB}, for $\Lambda$ to be $n$-PIB we require that for all measurement assemblages $N_{a|x}$, there exists an $n$-PEB channel $\Omega$ and a set of measurements $M_{a|x}$ such that
\be 
\Lambda^*(N_{a|x})=\Omega^*(M_{a|x}). \label{eq:explicit-n-PIB}
\ee
Therefore, by simply taking $\Omega:=\Lambda$ and $M_{a|x}:=N_{a|x}$ in Eq.~\ref{eq:explicit-n-PIB}, we immediately arrive at the following result:

\begin{proposition}
Every $n$-PEB channel is $n$-PIB.
\end{proposition}

It is illuminating to consider the corresponding Choi states. For $n$-PEB channels, the Choi states are exactly the states with Schmidt number $n$. \cite{chruscinski2006partially}. For $n$-PIB channels, we have the following result:
\begin{theorem} \label{thm:pib = sdi-sn}
$\Lambda$ is $n$-PIB if and only if $\rho_\Lambda$ only leads to $n$-preparable assemblages.
\end{theorem}
\begin{proof}
Let $\sigma=\text{Tr}_A(\rho_\Lambda)$ fix the channel-state correspondence.
Suppose $\Lambda$ is $n$-PIB, that is, for all measurements $N_{a|x}$, we have that $\Lambda^*(N_{a|x})$ is $n$-simulable. By Theorem \ref{thm: n-sim = n-prep}, this is equivalent to $\sigma^{\frac{1}{2}}\Lambda^*(N_{a|x})^T\sigma^{\frac{1}{2}}$ being $n$-preparable for all $N_{a|x}$. Via channel-state duality, this is equivalent to 
\be
\text{Tr}_A(N_{a|x} \otimes \id \rho)
\ee
being $n$-preparable for all $N_{a|x}$.
\end{proof}

The result of the above Theorem is put into context of other similar type connections between a channel and its Choi state in Table~\ref{tab:cs-duality}. We note that our results on bounding entanglement assisted simulation models for the noisy singlet state translate directly into bounds on the identity channel under depolarising noise for being $n$-PIB on the resrticted class of projective measurements. This also shows that when only projective measurements are considered, there are channels that are $n$-PEB without being $n$-PIB.

\setlength{\tabcolsep}{5pt}
\renewcommand{\arraystretch}{2}
\begin{table}[]
    \centering
    \begin{tabular}{|c|c|c|}\hline
    Channel & State & Reference \\\hline
    Entanglement breaking & Separable & \cite{horodecki2003entanglement} \\
    Incompatibility breaking & Unsteerable & \cite{heinosaari2015incompatibility, kiukas2017continuous}\\
         $n$-PEB & SN $n$ & \cite{terhal2000schmidt, chruscinski2006partially} \\
         $n$-PIB & SDI-SN $n$ & Theorem \ref{thm:pib = sdi-sn}. \\\hline
    \end{tabular}
    \caption{Connections between channels and their Choi states. Our work naturally extends this picture by generalising both incompatibility breaking channels and unsteerable states in terms of dimension, and proving that they directly correspond to each other through generalised channel state-duality.}
    \label{tab:cs-duality}
\end{table}

\section{Conclusions}

We have uncovered deep connections between high-dimensional versions of quantum steering, measurement incompatibility, and quantum channels, and demonstrated how a rich transfer of information is possible between these areas. In particular, we showed that the concept of $n$-simulability for sets of POVMs is equivalent to $n$-preparability for state assemblages in steering. This generalises the well-known connection between steering and joint measurability, which simply corresponds here to the case $n=1$.

We identified the resources required for observing GHDS, in particular that both high-dimensional measurements and high-dimensional entanglement are necessary. In the light of these results, we conclude that the experiment of Ref.~\cite{designolle2021genuine} also demonstrates measurements in pairs of MUBs that are highly incompatible, in the sense that are they not $14$-simulable.

Another direction is the idea of quantifying the degree of steering of an entangled state via a dimension. We obtained optimal models for isotropic entangled state, considering all projective measurements. This can be seen as a generalisation of the well-known type of local (hidden state) models by Werner, now allowing for low-dimensional entanglement as a resource. In turn, this leads to a characterisation of channels that map any set of projective measurements into $n$-simulable ones.

There are many exciting notions to explore that would extend this research direction. It would be useful to have better bounds on both $n$-preparability and $n$-simulability, and our work demonstrates that any progress here can be readily applied to both notions, providing a practical bridge between the two scenarios. Of particular interest would be to find bounds on the isotropic state being of SDI-SN $n$ under all POVMs, which would directly translate into the $n$-simulability of all POVMs. This follows analogous lines to the $n=1$ case (finding LHS bounds under projective/POVM measurements) \cite{barrett2002nonsequential}.

A natural further question would be to explore these questions in the context of nonlocality \cite{brunner2014bell}, which can be thought of as a fully-device independent (FDI) regime. Analogously to the steering case, one could define a behaviour $p(a,b|x,y)$ to be $n$-preparable if it could have arisen from a shared state of Schmidt number at most $n$, and define a state to have fully-device independent Schmidt number $n$ (FDI-SN $n$) if it can only lead to $n$-preparable behaviours. This is related to \cite{brunner2008testing}, where the authors introduce the concept of dimension witnesses to lower bound the dimension of the underlying state. One can quickly see in this scenario that if either of the two parties use $n$-simulable measurements, then the resulting behaviour will be $n$-preparable. Similarly, uncharacterised measurements on an $n$-preparable assemblage can only result in an $n$-preparable behaviour. However, it is less clear how one could characterise the corresponding channels whose Choi states have FDI-SN $n$. In the steering case we were able to exploit and generalise known connections with measurement incompatibility, but it seems that new tools may be needed to attack this problem in the fully device independent regime.

\textit{Acknowledgments.---} We acknowledge financial support from the Swiss National Science Foundation (projects 192244, Ambizione PZ00P2-202179, and NCCR SwissMAP). BDMJ acknowledges support
from UK EPSRC (EP/SO23607/1). T.C. would like to acknowledge the funding Deutsche Forschungsgemeinschaft (DFG, German Research Foundation) under Germany's
Excellence Strategy EXC-2123 QuantumFrontiers 390837967, as well as the support of the Quantum Valley Lower Saxony and the DFG through SFB 1227 (DQ-mat).

\bibliographystyle{unsrt}
\bibliography{references}

\appendix

\section{Proof of Proposition \ref{prop:n-PEB equivalence}}

\begin{proposition} \label{prop:n-PEB equivalence} The following are equivalent:
\begin{enumerate}[(i)]
    \item $\Lambda$ is $n$-PEB.
    \item SN$(\rho_\Lambda) \leq n$ for any choice of the marginal state in the generalised channel-state duality (see \cite{kiukas2017continuous}).
    \item There exists a Kraus decomposition of $\Lambda (Y) = \sum_\lambda K_\lambda Y K_\lambda^\dagger $ such that rank$(K_\lambda) \leq n$ for all $\lambda$.
\end{enumerate}
\end{proposition}

\begin{proof}The implication $(i) \Longrightarrow (ii)$ is immediate, and the equivalence $(i) \Longleftrightarrow (iii)$ is proven in \cite{chruscinski2006partially}.

To show $(ii) \Longrightarrow (iii)$, first recall that as highlighted in \cite{chruscinski2006partially}, any bipartite pure state $\ket{\psi} = \sum_{ij} \psi_{ij} \ket{i}\ket{j}$ can be written in the form $\ket{\psi}=\sum_{i} \ket{i} F\ket{i}$, where $\bra{j}F\ket{i} = \psi_{ij}$ and the Schmidt rank SR$(\ket{\psi})=\text{rank}(F)$. This generalises to mixed states as
\begin{align}
\rho &= \sum p_\lambda \ketbra{\psi_k}{\psi_k} \\
&= \sum_\lambda p_\lambda \sum_{ij} \ketbra{i}{j} \otimes F_\lambda \ketbra{i}{j} F_\lambda^\dagger
\end{align}
where $\text{Tr}(F_\lambda F_\lambda^\dagger)=1$ for all $\lambda$. One can also take $\text{rank}(F_\lambda) \leq \text{SN}(\rho)=n \quad  \forall \lambda$. (again see \cite{chruscinski2006partially}).
Now consider 
\begin{align}
 &\Lambda_\rho(Y) = \text{Tr}_B ( \mathbbm{1}\otimes (\sigma^{\tiny -\frac{1}{2}} Y \sigma^{\tiny -\frac{1}{2}} )^T \rho) \\
 &= \sum_{ij\lambda} p_\lambda \text{Tr}_B \bigg ( \mathbbm{1}\otimes (\sigma^{\tiny -\frac{1}{2}} Y \sigma^{\tiny -\frac{1}{2}} )^T  \ketbra{i}{j} \otimes F_\lambda \ketbra{i}{j} F_\lambda^\dagger \bigg ) \\
 &= \sum_{ij\lambda} p_\lambda \text{Tr}_B \bigg ( \mathbbm{1}\otimes F_\lambda^\dagger (\sigma^{\tiny -\frac{1}{2}} Y \sigma^{\tiny -\frac{1}{2}} )^T F_\lambda  \ketbra{i}{j} \otimes  \ketbra{i}{j} \bigg ) \\
 &= \sum_{\lambda} p_\lambda \bigg (F_\lambda^\dagger (\sigma^{\tiny -\frac{1}{2}} Y \sigma^{\tiny -\frac{1}{2}} )^T F_\lambda \bigg )^T 
\end{align}
where we used the well known fact $\sum_{ij}\text{Tr}_B (X\otimes Y \ketbra{i}{j} \otimes \ketbra{i}{j}) = XY^T$. Using properties of the transpose, we can now write
\begin{align}
  \Lambda_\rho(Y) &= \sum_{\lambda} p_\lambda \bigg (F_\lambda^T \sigma^{\tiny -\frac{1}{2}} Y \sigma^{\tiny -\frac{1}{2}}  (F_\lambda^\dagger)^T \bigg ) \\
  &\equiv \sum_\lambda K_\lambda ~ Y ~ K_\lambda^\dagger
\end{align}
where we defined $K_\lambda := \sqrt{p_\lambda} F_\lambda^T \sigma^{-\frac{1}{2}}$. As the channel $\Lambda_\rho$ is trace-preserving, we must have that 
\be
\text{Tr} \bigg (\sum_\lambda K_\lambda^\dagger K_\lambda Y \bigg ) = \text{Tr} (Y) \qquad \forall ~ Y
\ee
which implies that $\sum_\lambda K_\lambda^\dagger K_\lambda = \id$, and hence $K_\lambda$ define a valid set of Kraus operators for the channel.
Finally, observe that
\begin{align}
\text{rank}(K_\lambda) &= \min\{ \text{rank}(F_\lambda^T), \text{rank}( \sigma^{-\frac{1}{2}})\} \\
&= \text{rank}(F_\lambda^T) = \text{rank}(F_\lambda) \leq n
\end{align}
for all $\lambda$, which completes the proof.
\end{proof}

\section{Proof of Theorem \ref{thm:weight}}

Recall the definition of the following sets:
\begin{itemize}
    \item $E_n$: the set of states with Schmidt number at most $n$.
    \item $S_n$: the set of of $n$-simulable measurement assemblages.
    \item $P_n$: the set of $n$-preparable steering assemblages.
\end{itemize}

We now prove the following theorem, which serves to quantify the necessity of both high-dimensional incompatibility and entanglement for genuine high-dimensional steering.

\weighttheorem*

\begin{proof}
We first note that the sets $E_n$, $S_n$ and $P_n$ are all convex, which can be readily verified (see for example Refs. \cite{terhal2000schmidt,uola2019quantifying}).

\begin{align}
    \sigma_{a|x}&=\text{Tr}_A[(M_{a|x}\otimes\id)\rho_{AB}]\\
    &=(1-\mathcal{W}_{S_n}(M_{a|x}))\mathcal{W}_{E_n}(\rho_{AB})]\tau^{(1)}_{a|x}\nonumber\\
        &+\mathcal{W}_{S_n}(M_{a|x})(1-\mathcal{W}_{E_n}(\rho_{AB}))]\tau^{(2)}_{a|x}\nonumber\\
        &+(1-\mathcal{W}_{S_n}(M_{a|x}))(1-\mathcal{W}_{E_n}(\rho_{AB}))]\tau^{(3)}_{a|x}\nonumber\\
    &+\mathcal{W}_{S_n}(M_{a|x})\mathcal{W}_{E_n}(\rho_{AB})\kappa_{a|x},
\end{align}
where $\tau^{(i)}_{a|x}$ is an $n$-preparable state assemblage for $i=1,2,3$ and $\kappa_{a|x}$ is an arbitrary assemblage. The fact that $\tau^{(i)}_{a|x}$ are all $n$-preparable follows directly from Theorem \ref{theorem:compat->prep}. Now note that the coefficients of the first three terms sum to $1-\mathcal{W}_{S_n}(M_{a|x})\mathcal{W}_{E_n}(\rho_{AB})$, hence we can write the sum of these first three terms as
\[
(1-\mathcal{W}_{S_n}(M_{a|x})\mathcal{W}_{E_n}(\rho_{AB}))\tau_{a|x}
\]
where $\tau_{a|x}$ is a convex combination of $\tau^{(1)}_{a|x}$, $\tau^{(2)}_{a|x}$ and $\tau^{(3)}_{a|x}$, and hence is itself $n$-preparable. Putting this together, we have that
\begin{align}
  \sigma_{a|x}&=(1-\mathcal{W}_{S_n}(M_{a|x})\mathcal{W}_{E_n}(\rho_{AB}))\tau_{a|x}\\
    &+\mathcal{W}_{S_n}(M_{a|x})\mathcal{W}_{E_n}(\rho_{AB})\kappa_{a|x},
\end{align}
so we see that $\mathcal{W}_{S_n}(M_{a|x})\mathcal{W}_{E_n}$ is a feasible solution for the convex weight of $\sigma_{a|x}$ with respect to $P_n$. As the convex weight is a minimisation the inequality
\[
\mathcal{W}_{P_n}(\sigma_{a|x}) \leq \mathcal{W}_{S_n}(M_{a|x}) \mathcal{W}_{E_n}(\rho_{AB}).
\]
follows.
\end{proof}

\end{document}